\documentclass{birkjour}
 \newtheorem{thm}{Theorem}[section]

 \theoremstyle{definition}
 
 \theoremstyle{remark}
 \newtheorem{rem}[thm]{Remark}
 \newtheorem*{ex}{Example}
 \numberwithin{equation}{section}

\def\o#1{\ensuremath \mathcal{#1}}

\begin{document}

%
%
%
%
%
%
%
%
%
\title[Darboux Transformations]
 {Invariants for Darboux transformations \\
 of Arbitrary Order for \\
  $D_x D_y +aD_x + bD_y +c$}
\author[Ekaterina Shemyakova]{Ekaterina Shemyakova}

\address{%
Department of Mathematics\\
SUNY at New Paltz\\
1 Hawk Dr. New Paltz, NY 12561}

\email{shemyake@newpaltz.edu}



\begin{abstract}
We develop the method of regularized moving frames of Fels and Olver to obtain explicit
general formulas for the basis invariants that generate 
all the joint differential invariants, under gauge transformations, for the operators
\[
\o{L}=D_xD_y +a(x,y) D_x + b(x,y) D_y +c(x,y)
\]
and an operator of arbitrary order.

The problem appeared in connection with invariant construction of Darboux transformations for $\o{L}$.
\end{abstract}

\maketitle
\section{Introduction}
The present paper is devoted to Darboux transformations~\cite{ts:steklov_etc:00} for the Laplace operator,
\begin{equation} \label{op:L2}
\o{L} = D_xD_y +a D_x + bD_y +c \  ,
\end{equation}
where $a,b,c \in K$, where $K$ is some differential field (see Sec.~\ref{sec:pre}). Operator $\o{L}$ is transformed into operator $\o{L}_1$ with the same principal symbol
 (see Sec.~\ref{sec:pre}) by means of operator $\o{M}$  if there is a linear partial differential operator $\o{N}$ such that
\begin{equation} \label{main}
\o{N} \o{L} = \o{L}_1 \o{M} \  .
\end{equation}
In this case we shall say that there is a \textit{Darboux transformation for pair} $(\o{L},\o{M})$;
we also say that $\o{L}$ \textit{admits a Darboux transformation generated by} $\o{M}$.
We define \textit{the order of a Darboux transformation} as the order of the $\o{M}$ corresponding to it.

Given some operator $\o{R} \in K[D]$ and an invertible function $g \in K$, the corresponding
\textit{gauge transformation} is defined as
\[
 \o{R} \to \o{R}^{g} \ , \ \o{R}^g = g^{-1}  R g \  .
\]
The principal symbol of an operator in $K[D]$ is invariant under the gauge transformations.
One can prove~\cite{shem:darboux2}  that if a Darboux transformation exists for a pair $(\o{L},\o{M})$,
then a Darboux transformation exists for the pair $(\o{L}^g,\o{M}^g)$.
Therefore, Darboux transformations can be considered  for the equivalence classes of the pairs~$(\o{L},\o{M})$.

A function of the coefficients of an operator in $K[D]$ and of the derivatives of the coefficients of the operator is called a \textit{differential invariant}
with respect to the gauge transformations if it is unaltered under the action of the gauge transformations.  For several operators,
we can consider \textit{joint differential invariants}, which are functions of all their coefficients and of the derivatives of these coefficients.
Differential invariants form a differential algebra over $K$. This algebra may be $D$-generated over $K$ by some number of basis invariants.
We say that these basis invariants form a \textit{generating set of invariants}.  
For operators  of the form~\eqref{op:L2} known as \textit{Laplace invariants}
functions $k=b_y + ab -c$ and $h=a_x + ab -c$ form a generating set of invariants.

In~\cite{movingframes} we developed the regularized moving
frames of Fels and Olver~\cite{FO2,olver2011diff_inv_algebras,Mansf_book} for the individual linear partial differential operators of orders $2$ and $3$ on the plane under
gauge transformations and obtained generating sets of invariants for those operators.
In the present paper we extend those ideas and show that there is a finite generating set of invariants
for the pairs~$(\o{L},\o{M})$, where $\o{L}$ is of the form~\eqref{op:L2} and $\o{M} \in K[D]$ is of arbitrary form and of order $d$.
For $\o{M}$ of arbitrary order $d$ but given in its normalized form without mixed derivatives, we find explicit general formulas for the basis invariants
for the pairs~$(\o{L},\o{M})$ (Theorem~\ref{thm:1}).

The existence of such normalized forms for $\o{M}$ is implied by one of the theorems proved in~\cite{shem:darboux2},
which can be re-formulated as follows.
\begin{thm} \cite{shem:darboux2}
Let there be a Darboux transformation for pair $(\o{L},\o{M})$, where
$\o{L}$ be of the form (\ref{op:L2}) and $\o{M} \in K[D]$ of arbitrary form and order.
Then there is a Darboux transformation for pair $(\o{L},\o{M}')$, where
$\o{M}'$ contains no mixed derivatives.
\end{thm}
\section{Preliminaries}
\label{sec:pre}
Let $K$ be a differential field of characteristic zero, equipped with commuting derivations $\partial_x, \partial_y$.
Let $K[D]=K[D_x, D_y]$ be the corresponding ring of linear
partial differential operators over $K$, where $D_x, D_y$ correspond to derivations $\partial_x, \partial_y$.
One can either assume field $K$ to be differentially closed, in other words containing all the solutions
of, in general nonlinear, Partial Differential Equations (PDEs)
with coefficients in $K$, or simply assume that $K$ contains the solutions of those PDEs that we encounter on the way.

Let $f \in K$, and $\o{L} \in K[D]$; by $\o{L} f$ we denote the composition of operator $\o{L}$ with the operator of multiplication
by a function $f$, while  $\o{L} (f)$ mean the application of operator $\o{L}$ to $f$.
 The second lower index attached to a symbol denoting a function means the derivative of that function with respect to the variables listed there.
For example, $f_{1,xyy} = \partial_x \partial_x \partial_y f_1$.

In the present paper we use Bell polynomials,
\[
\begin{aligned}
&B_{n,k}(x_1,x_2,\dots,x_{n-k+1}) = \\
&                               \sum
\frac{n!}{j_1!j_2!\cdots j_{n-k+1}!} \left( \frac{x_1}{1!}\right)^{j_1} \left( \frac{x_2}{2!}\right)^{j_2}\cdots\left( \frac{x_{n-k+1}}{(n-k+1)!}\right)^{j_{n-k+1}} \ ,
\end{aligned}
\]
where the sum is over all sequences $j_1, j_2, j_3, \dots, j_{n-k+1}$ of non-negative integers such that
$j_1+j_2+\cdots = k$ and $j_1+2j_2+3j_3+\cdots=n$.
The sum
\[
 B_n(x_1,\dots,x_n)=\sum_{k=1}^n B_{n,k}(x_1,x_2,\dots,x_{n-k+1})
\]
is called the $n$-th \textit{complete Bell polynomial}, and also it has the following determinant representation:
\[
    B_n(x_1,\dots,x_n) = \det\begin{bmatrix}x_1 & {n-1 \choose 1} x_2 & {n-1 \choose 2}x_3 & {n-1 \choose 3} x_4  & \cdots & \cdots & x_n \\ \\
-1 & x_1 & {n-2 \choose 1} x_2 & {n-2 \choose 2} x_3  & \cdots & \cdots & x_{n-1} \\
\\ 0 & -1 & x_1 & {n-3 \choose 1} x_2  & \cdots & \cdots & x_{n-2} \\ \\
\vdots & \vdots & \vdots &  \vdots & \ddots & \ddots & \vdots \\
0 & 0 & 0  & 0 & \cdots & -1 & x_1 \end{bmatrix}.
\]
\section{Invariants for  Normalized Darboux Transformations of Arbitrary Order}
\begin{thm} \label{thm:1} All joint differential 
invariants\footnote{with respect to gauge transformations}
for the pairs $(\o{L}, \o{M})$, where 
$\o{L} = D_xD_y +a D_x + bD_y +c$ and 
${ \displaystyle \o{M}=\sum_{i=1}^{d} m_i D_x^i +m_{-i} D^i_y}+m_0$
and where $m_i \in K, i=-d, \dots, d$ and $a,b,c \in K$, can be generated by the following $2d+3$ basis invariants.
\begin{align*}
 m     &= a_x - b_y \ ,    \\
 h     &= ab-c + a_x \ ,   \\
 R_j   &= \sum_{w=j}^d m_w \binom{w}{j} B_{w-j}(-b, -\partial_x(b), - \partial_x^2(b), \dots, -\partial_x^{w-j-1}(b) ) \ , \\
 R_{-j}&= \sum_{w=j}^d m_{-w} \binom{w}{j} B_{w-j}(-a, -\partial_y(a), - \partial_y^2(a), \dots, -\partial_y^{w-j-1}(a) ) \ , \\
 R_0   &= \sum_{w=1}^d m_w B_{w}(-b, -\partial_x(b), - \partial_x^2(b), \dots, -\partial_x^{w-1}(b) ) \\
       & +   m_{-w} B_{w}(-a, -\partial_y(a),
 - \partial_y^2(a), \dots, -\partial_y^{w-1}(a) ) +m_0  \ .
\end{align*}
\end{thm}
\begin{proof}
We adopt the method of regularized moving frames~\cite{FO2,olver2011diff_inv_algebras,Mansf_book}.
Possible difficulties with the infinite dimensional case are addressed in~\cite{OP:2008:inf_dim}.
In this short paper we refer the reader to these works for the rigorous notation and for a justification of the method.

For transformations $\o{L} \mapsto \o{L}^{\exp(\alpha)}$, which implies the following group action on the coefficients of the operator,
\[
 (a,b,c,) \to (a + \alpha_y , b + \alpha_x , c + a\alpha_x + b\alpha_y + \alpha_{xy} + \alpha_x \alpha_y)  
 \  ,
\]
consider the prolonged action.

Let us construct a frame
\[
 \rho: (a_J, b_J, c_J) \mapsto g \  ,
\]
at some regular point $(x^0,y^0)$.  
Here $a_J$ denotes the jet coefficients of $a$ at $(x^0,y^0)$,
and they are to be regarded as the independent group parameters.
A moving frame can be constructed through a normalization procedure based on a choice of a cross-section
to the group orbits. Here we define a cross-section  by normalization equations
\begin{align*}
 &a_{J} = 0 \ ,  \\
 &b_{X} = 0 \ ,
\end{align*}
where here $J$ is a string of the form $x \dots xy \dots y$, where $y$ has to be present at least once, and
there may be no $x$-s, and $X$ is a string of the form $x \dots x$, where there can be no $y$-s.
The normalization equations when solved for group parameters produces the moving frame section:
\begin{align*}
 & a_J = a_{J-y} \ ,  \\
 & b_X = b_{X-x} \ ,
\end{align*}
where $J-y$ means that we take one $y$ from the string $J$, and $X-x$ means that we take one $x$ from the string $X$.
The first two fundamental differential invariants can be then found:
\begin{align*}
&(b_1)_y \Big|_{\rho} =b_y + \alpha_{xy} = b_y - a_x = m = h-k \  , \\
& c_1 \Big|_{\rho}    =c - ab - ab -a_x +ab = c -a_x -ab = h \  .
\end{align*}

The remaining invariants of the generating set can be obtained using the constructed frame
for the group acting on all the coefficients of the second operator in the pair,  operator $\o{M}$
since none of them has been used during the normalization process and construction of the frame.

Consider a gauge transformation of each of the terms in the sum $\o{M}=\sum_{i=1}^{d} m_i D_x^i +m_{-i} D^i_y + m_0$:
\begin{align}
& \left(m_i D_x^i \right)^{\exp(\alpha)} =  \exp(-\alpha) \cdot m_i \cdot D_x^i \circ  \exp(\alpha) \ , \; i \neq 0 \ ,  \label{eq:mi_conj} \\
& \left(m_{-i} D_y^i \right)^{\exp(\alpha)} =  \exp(-\alpha) \cdot m_{-i} \cdot D_y^i \circ  \exp(\alpha) \ , \; i \neq 0 \ ,  \nonumber \\
& \left(m_0 \right)^{\exp(\alpha)} = m_0 \ . \nonumber
\end{align}
Using the General Leibnitz rule~\cite{olver1986applications}, for $i \neq 0$, we have
\[
 \left(m_i D_x^i \right)^{\exp(\alpha)} =  \exp(-\alpha) \cdot m_i \cdot \sum_{k=0}^i \binom{i}{k} \frac{\partial^{i-k}
                                      \exp(\alpha)}{\partial x^{i-k}} \cdot D^k_x \  ,
\]
then  applying Fa\`{a} di Bruno formula~\cite{stanley2000enumerative_combinatorics}  we continue
\begin{align*}
     \left(m_i D_x^i \right)^{\exp(\alpha)} =  & \exp(-\alpha) \cdot m_i \cdot \sum_{k=0}^i \binom{i}{k}
                     \sum_{t=1}^{i-k} \frac{\partial^{t}  \exp(\alpha)}{\partial \alpha^{t}} \\
                      & \cdot B_{i-k,t} \left(\partial_x(\alpha), \partial_x^2(\alpha), \dots, \partial_x^{i-k-t+1}(\alpha) \right)     \cdot D^k_x \\
                                        =& m_i \cdot \sum_{k=0}^i \binom{i}{k} \sum_{t=1}^{i-k} B_{i-k,t} (\partial_x(\alpha), \partial_x^2(\alpha), \dots, \partial_x^{i-k-t+1}(\alpha) )    \cdot D^k_x  \ ,
\end{align*}
where $B_{i-k,t}$ are Bell polynomials.
Since only the terms $B_{i-k,t}$ are summed with respect to $t$,  we  can rewrite the expression in terms of complete Bell polynomials:
\begin{equation}
\label{eq:1}
 \left(m_i D_x^i \right)^{\exp(\alpha)} = m_i \cdot \sum_{k=0}^i \binom{i}{k} B_{i-k} (\partial_x(\alpha), \partial_x^2(\alpha), \dots, \partial_x^{i-k}(\alpha) )    \cdot D^k_x  \  .
\end{equation}
Now we compute invariants $R_j$, $j = 1, \dots, d$ from the statement of the theorem as the coefficients at $D^j_x$,
restricted on the constructed frame $\rho$.
First, equality~\eqref{eq:1} implies that $m_i$ appears only in $R_j$ with $j \leq i$.
Secondly, equality~\eqref{eq:1} implies that $R_j$ must be a sum of the $m_i$ multiplied by some functional coefficients. The coefficient of
$m_i$ in $R_j$ can be found from~\eqref{eq:1} by the substitution $k=j$. 
In this way, we can obtain the invariants
\begin{align*}
R_j=&\sum_{w=j}^d m_w \binom{w}{j} B_{w-j}(\partial_x(\alpha), \partial_x^2(\alpha), \dots, \partial_x^{w-j}(\alpha) ) \Big|_{frame} =  \\
     =&\sum_{w=j}^d m_w \binom{w}{j} B_{w-j}(-b, -\partial_x(b), - \partial_x^2(b), \dots, -\partial_x^{w-j-1}(b) )
\end{align*}
for $j=1, \dots, d$.
Analogously, one can compute invariants $R_{-j}$, $j = 1, \dots, d$ from the statement of the theorem as the coefficients at $D^j_y$,
restricted on the constructed frame $\rho$.

We compute invariant  $R_0$  from the statement of the theorem as function $\o{M}(1)$ restricted on the constructed frame $\rho$.
It has to be considered separately as this is the only invariant which contains both $a$ and $b$.
\begin{align*}
R_0=&\sum_{w=1}^d m_w \binom{w}{0} B_{w}(\partial_x(\alpha), \partial_x^2(\alpha), \dots, \partial_x^{w}(\alpha) ) \Big|_{frame} \\
    &+ \sum_{w=1}^d m_w \binom{w}{0} B_{w}(\partial_y(\alpha), \partial_y^2(\alpha), \dots, \partial_y^{w}(\alpha) ) \Big|_{frame} + m_0 =  \\
   =&\sum_{w=1}^d m_w B_{w}(-b, -\partial_x(b), - \partial_x^2(b), \dots, -\partial_x^{w-1}(b) ) \\
    &+   m_{-w} B_{w}(-a, -\partial_y(a), - \partial_y^2(a), \dots, -\partial_y^{w-1}(a) ) + m_0  \  .
\end{align*}
\end{proof}
\begin{thm}[Alternative form of Theorem~\ref{thm:1}] All joint differential invariants\footnote{with respect to gauge transformations}
for the pairs $(\o{M}, \o{L})$, where ${ \displaystyle \o{M}=\sum_{i=1}^{d} m_i D_x^i +m_{-i} D^i_y}+m_0$ and
$\o{L} = D_xD_y +a D_x + bD_y +c$, where $m_i \in K, i=-d, \dots, d$ and $a,b,c \in K$, can be generated by the following $2d+3$ basis invariants.
\begin{align*}
 m     &= a_x - b_y\,,    \\
 h     &= ab-c + a_x\,,   \\
 R_0   &= {\displaystyle m_0 +
          \sum_{i=1}^{d} \left( m_i P_i(b)+m_{-i}P_i(a)\right) }\,, \\
 R_j   &= {\displaystyle \sum_{i=0}^{d-j} \binom{j+i}{j} m_{i+j} P_i(b) \ , \ j \geq 1 }\,, \\
 R_{-j}&= {\displaystyle \sum_{i=0}^{d-j} \binom{j+i}{j} m_{-(i+j)} P_i(a) \ , \ j \geq 1 }\,,
\end{align*}
where
\begin{align*}
 P_0(f)&= 1 \ ,          \\
 P_i(f)&= -\Omega^i(f) \ , i \in \mathbb{N}_0 \ ,
\end{align*}
and the linear differential operator $\Omega$ is defined by
\[
 \Omega (f)=
\left\{
\begin{aligned}
& (D_x - b)(f) \ , \; \text{if} \quad f=b \\
& (D_y - a)(f) \ , \; \text{if} \quad f=a \ .
\end{aligned}
\right.
\]
\end{thm}
\begin{rem} Note the explicit forms taken by the first few operators $P_i$:
\begin{align*}
 P_0(f)&= 1\,,          \\
 P_1(f)&= -f\,,          \\
 P_2(f)&= -f_x + f^2 \ ,    \\
 P_3(f)&= -f_{xx} + 3ff_x -f^3\,,          \\
 P_4(f)&= -f_{xxx} + 4ff_{xx} - 6f_xf^2 +3f_x^2 + f^4 \ ,\\
 P_5(f)&= -f_{xxxx}+ 5 f f_{xxx} + 10 f_{xx} f_x - 15 f_x^2 f - 10 f_{xx} f^2 + 10 f_x f^3  - f^5 \ .
\end{align*}
\end{rem}
 \begin{ex} 
Given the operator ${ \displaystyle \o{M}=\sum_{i=1}^{5} m_i D_x^i +m_{-i} D^i_y}$ and 
an operator $\o{L}$
in the form~\eqref{op:L2}, with $m_i \in K, i=-5, \dots, 5$, the following functions form a generating set
of invariants.
\begin{align*}
m     &= a_x - b_y \ ,    \\
 h     &= ab-c + a_x \ ,   \\
R_1 & =-2m_2b+m_1+(-3b_x+3b^2)m_3+(-4b_{xx}+12b_xb-4b^3)m_4 \\
       & +(20b_{xx}b+15b_x^2-30b_xb^2-5b_{xxx}+5b^4)m_5 \ , \\
R_2 & = m_2-3m_3b+(-6b_x+6b^2)m_4+(-10b_{xx}-10b^3+30b_xb)m_5 \ , \\
R_3 & = m_{3}-4m_{4}b+(-10b_{x}+10b^2)m_{5} \ , \\
R_4 & = m_4-5m_5b \ , \\
R_5 & = m_{5} \ , \\
R_{-5} & = m_{-5} \ , \\
R_{-4} & = m_{-4}-5m_{-5}a \ , \\
R_{-3} & = m_{-3}-4m_{-4}a+(-10a_{y}+10a^2)m_{-5} \ , \\
R_{-2} & = m_{-2}-3m_{-3}a+(-6a_{y}+6a^2)m_{-4}+(-10a_{yy}-10a^3+30a_{y}a)m_{-5} \ , \\
R_{-1} & = -2m_{-2}a + m_{-1}+(-3a_y+3a^2)m_{-3}+(-4 a_{yy}+12a_ya - 4 a^3) m_{-4} \\
            &+ (20 a_{yy}a+15a_y^2-30a_y a^2- 5 a_{yyy}+5 a^4)m_{-5}  \ .
\end{align*}
and finally
\newpage 
\begin{align*}
R_0 = & m_0 - b m_1 -  a m_{-1} +(-b_x+b^2)m_2+(a^2-a_{y})m_{-2}
       +(3b_xb-b^3-b_{xx})m_3 \\
      & +(3a_{y}a-a_{yy}-a^3)m_{-3} +(b^4+3b_x^2-b_{xxx}+4b_{xx}b-6b_xb^2)m_4 \\
      & +(-6a_{y}a^2+a^4-a_{yyy}+3a_{y}^2+4a_{yy}a)m_{-4}\\
      & +(5b_{xxx}b-b_{xxxx}-b^5-10b_{xx}b^2+10b_xb^3-15b_x^2b+10b_{xx}b_x)m_5\\
      & +(-15a_{y}^2a+10a_{yy}a_{y}+5a_{yyy}a-a_{yyyy}-a^5-10a_{yy}a^2+10a_{y}a^3)m_{-5}
\ .\end{align*}
\end{ex}

\end{document}